\documentclass[11pt]{amsart}   % For Latex2e
\usepackage{amssymb,xy}
\usepackage{amsthm, amsmath}
\usepackage[latin1]{inputenc}
\usepackage[T1]{fontenc}
\xyoption{all}
\theoremstyle{plain}
\newtheorem{theorem}{Theorem}[section]
\newtheorem{lemma}[theorem]{Lemma}
\newtheorem{proposition}[theorem]{Proposition}

\theoremstyle{definition}
\newtheorem{definition}[theorem]{Definition}

\theoremstyle{remark}
\title{Self-Dual codes  from $(-1,1)$-matrices of skew symmetric type}
\author{
Jos\'e Andr\'es Armario
\and
María Dolores Frau
}
\date{}

\address{\small \rm  Depto Matemática Aplicada I\\
 Universidad de Sevilla\\ Avda. Reina Mercedes s/n 41012 Sevilla\\ Spain.}
\email{  armario@us.es}

\email{mdfrau @us.es}

\begin{document}

\begin{abstract}
Previously, self-dual codes have  been constructed from weighing matrices, and in particular from conference matrices (skew and symmetric). In this paper, codes constructed from  matrices of skew symmetric type whose determinants reach the Ehlich-Wojtas' bound are presented. A necessary and sufficient condition for these codes to be self-dual is given, and examples are provided for lengths up to 52.
\end{abstract}
 %\PACS{05B20 \and 94B25 \and 94B05}

\maketitle

A linear $[N,K]$ code $C$ over $GF(p)$ is a $K$-dimensional vector subspace of $GF(p)^{N}$, where $GF(p)$ is the Galois field with $p$ elements, $p$ prime. The elements of $C$ are called codewords and the weight  $wt(x)$ of a codeword $x$ is the number of nonzero coordinates in $x$. The minimum weight (or minimum distance for linear codes) of $C$ is defined as $d(C)=\mbox{min}\,\{wt(x)\colon x\in C\,\mbox{and}\, x\neq 0\}$. The error-correction capability of a code has to do with its minimum distance. In this sense,  codes with  large minimum distance  among $[N,K]$ codes are the most desirable.
 A matrix whose rows generate the code $C$ is  called a generator matrix of $C$. The dual code $C^\perp$ of $C$ is defined as
 $C^\perp=\{x\in GF(p)^N\colon x\cdot y=0 \,\mbox{for all} \,y\in C\}$. $C$ is {\it self-dual} if $C=C^\perp$. We say that a self-dual code of distance $d$ is {\it optimal} if there is no other self-dual code of that length with a greater minimum weight.

 In \cite{GH99} a technique for constructing self-dual codes using weighing matrices is presented. In \cite{AG01}, focusing on skew weighing matrices,  an  improved method for constructing self-dual codes is given. Among  weighing matrices, conference matrices are important because they provide a large number of self-dual codes with the potential for high minimum distance. It is well known that if $n$ is the order of a conference matrix  then it has to be even. For $n= 2\mod 4$ conference matrices are (equivalent to) symmetric. Otherwise, they are skew symmetric.

In this paper, working with $n= 2\mod 4$ a method for constructing self-dual codes from $n$ by $n$ $(-1,1)$-matrices $K$ of skew type with the maximum determinant is presented. Let us point out that $K-I_n$ is a skew symmetric matrix with $0$'s on the main diagonal and $\pm 1$ elsewhere ($I_n$ denotes the identity matrix of order $n$). Furthermore,  by Theorem \ref{theorembound}  there is no other skew matrix with a greater determinant than $K-I_n$. These two properties are used in our construction. In this sense,  it follows the same philosophy of the method for skew conference matrices (Proposition \ref{sdcsw}).

\noindent{\bf Notation.} Throughout this paper we use $-$ for $-1$ and $1$ for $+1$. The notation $(0,1)$-matrix  means a matrix whose entries are either $0$ or  $1$. We use $J_n$ for the  matrix all whose entries are equal to one of order $n$ and $M^T$ for the transpose of $M$. Whenever a determinant is mentioned in this paper, we mean its absolute value.

\section{Two known construction methods}

In this section we recall two techniques for constructing self-dual codes using weighing matrices (see \cite{AG01,GH99}).

\subsection{Weighing and conference matrices}
Firstly, we need some preliminary notions and results.

A {\it weighing matrix} $W(n,k)$ of order $n$ and weight $k$ is a  square $(0,1,-1)$-matrix of size $n$ such that $W\cdot W^T=k I_n,\,k\leq n$. Obviously, the determinant of $W(n,k)$ is $k^{n/2}$. A weighing matrix $W(n,n)$ is called a {\it Hadamard matrix}, and must have order 1, 2, or a multiple of 4. A weighing matrix $W(n,n-1)$ is called a {\it conference matrix}.

A matrix $W$ is {\it symmetric} if $W=W^T$. A matrix is {\it skew-symmetric} (or {\it skew}) if $W=-W^T$.

\begin{theorem}\cite[p.497]{CD96}
If $W=W(n,n-1)$, then either $n=0 \mod  4$ and $W$ is equivalent to a skew matrix, or $n=2 \mod  4$ and $W$ is equivalent to a symmetric matrix and such a matrix cannot exist unless $n-1$ is the sum of two squares: thus they cannot exist for orders $22,34,58,70,78,94.$ The first values for which the existence of symmetric conference matrices is unknown are $n=66,86$.
\end{theorem}

For more details and construction of weighing matrices the reader can refer the book by Geramita and Seberry \cite{GS79}.

Two important properties of the weighing matrices, which follow directly from the definition are:
\begin{enumerate}
\item Every row and column of a $W(n,n-k)$ contains exactly $k$ zeros;
\item Every two distinct rows and columns of a $W(n,n-k)$ are orthogonal to each other, which means that their inner product is zero.
\end{enumerate}

For the determinant of skew symmetric matrices we have
\begin{lemma}\cite{How80}
\begin{enumerate}
\item If $n$ is odd and $A$ is a skew-symmetric matrix with real elements the $det (A)=0$.
\item If $n$ is even and $A$ a skew-symmetric matrix with real elements then $det(A)$ is $PF(A)^2$, where $PF(A)$ is the Pfaffian of $A$ a polynomial in the entries of $A$.
\end{enumerate}
\end{lemma}

Two matrices $M$ and $N$ are said to be {\it Hadamard equivalent} or {\it equivalent} if one can be obtained from the other by a sequence of the operations:
\begin{itemize}
\item interchange any pairs of rows and/or columns;
\item multiply any rows and/or columns through by $-1$.
\end{itemize}
In other words, we say that $M$ and $N$ are equivalent if there exist $(0,1,-1)$-monomial matrices $P$ and $Q$ such that $PMQ^T=N$.
\subsection{Self-dual codes using weighing matrices}
The following simple but powerful methods was introduced in \cite{AG01,GH99} and they have been extensively used in searching for optimal self-dual codes \cite{KS09}.

\newpage

\begin{proposition}\cite{GH99}\label{proposition1GH99}
Let $W(n,k)$ be a weighing matrix of order $n$ and weight $k$. Let $\alpha$ be a nonzero element of $GF(p)$ such that
$$\alpha^2 + k = 0 \mod p.$$
Then the matrix $$G=[\alpha I,\, W(n,k)]$$
generates a self-dual $[2n,n]$ code $C(n,k)$ over $GF(p)$.
\end{proposition}

The  method above relies on the orthogonality  of  rows and columns of $W(n,k)$.
Self-dual codes can also be constructed from skew weighing matrices using the following method where the skew-symmetric property of the matrices is used as well.

\begin{proposition}\cite{AG01}\label{sdcsw}
Let $W(n,k)$ be a skew weighing matrix of order $k$. Let $\alpha$ and $\beta$ be nonzero elements of $GF(p)$ such that
$$\alpha^2 +\beta^2 + k = 0 \mod p.$$
Then the matrix $$G=[\alpha I, \, \beta I + W(n,k)]$$ generates a self-dual code $C^*(n,k)$ over $GF(p)$.

\end{proposition}

Conference matrices are important because they provide a large number of self-dual codes with the potential for high minimum distances. We can find a good explanation for this, in the fact that conference matrices provide the maximal determinant amount the family of matrices with $0$'s on the main diagonal and $\pm 1$'s elsewhere.

Although the technique described in Proposition \ref{proposition1GH99} works for any weighing matrix (skew or not), let us point out  that this method is a particular  case ($\beta=0$) of the one given in Proposition \ref{sdcsw} when we are dealing with  skew weighing matrices.

For small $n$ and $p$, values of $\alpha$ providing self-dual codes with symmetric conference matrices and the distance of these code were given by Arasu and Gulliver in \cite{AG01}.  Table \ref{tabla1} summarized some of them.
%A blank in the table means that no self-dual code exists for this construction.
In this table and the one that follows, the column headed $d_b$ gives the known bounds on $d$ for a self-dual code from \cite{Gab2013}.
Rows with no values (or those which do not appear for $2<p<23$) indicate that a self-dual code cannot be constructed. The gaps in this table motivate the investigation of self-dual codes constructed from $(-1,1)$-matrices of skew type. These are given in the next section.

\begin{center}\begin{table}
\caption{Self-dual codes from symmetric conference matrices of orders $n=6$ and $14$.}
\label{tabla1}$$\begin{array}{c|c|c|c} \hline  & N=12 &  N=28  \\
\hline \begin{array}{c} p \\
\hline 3 \\ 5 \\ 7 \\ 11 \\ 13 \\ 17 \\ 19 \\ 23
\end{array}
 &
 \begin{array}{ccc}\quad\alpha \quad & d & \quad d_b\quad \\
\hline   1 & 6 & 6 \\
 &  & \\
3 & 6 & 6\\
 &  &  \\
 &  & \\
 &  &  \\
 &  & \\
8 & 6 & 7
\end{array}
&
\begin{array}{ccc}\quad\alpha \quad & d & \quad d_b\quad \\
\hline    &  & \\
 &   & \\
1 & 10 & \quad 11\,-\,13 \\
3 & 10 & \quad 10\,-\,14 \\
&  &  \\
2 & 10 & \quad 10\,-\,15\\
5 & 10 & \quad 10\,-\,15\\
 &  &
\end{array}
\\
 \hline \end{array}$$\end{table}\end{center}

\section{New construction method}

The codes generated using weighing matrices $W(n,k)$  for $k=n$  and $k=n-1$ (Hadamard and conference matrices) are often optimal or near  optimal. For smaller values of $k$, the distance of these codes become poorer and poorer.

For $n= 2 \mod 4$ no Hadamard matrix  exists, so the codes generated using conference matrices usually are the best option but  these matrices cannot be skew. So we must use  Proposition \ref{proposition1GH99} (Proposition \ref{sdcsw} is not possible) to generate these codes.
 Under this panorama, we have displayed a technique for constructing self-dual codes using skew matrices with the maximum possible determinant: $(2n-3)(n-3)^{\frac{1}{2}n-1}$.

Let us observe that $\det(W(n,k))=k^{n/2}$. The following result provides a relationship between  the determinant of weighing matrices and the maximum possible determinant for skew matrices whose off-diagonal elements are all $\pm 1$.
\begin{lemma}
Let $n\ge 4$ be an integer. Then,
$$k^{n/2}\leq  (2n-3)(n-3)^{\frac{1}{2}n-1} \leq (n-1)^{n/2},\,\,\, 1\leq k\leq n-2.$$
\end{lemma}

\subsection{Ehlich-Wotjas' bound for the determinant of skew matrices}

 Let $g(n)$  denote the maximum  determinant of all   $n\times n$  matrices with elements $\pm 1$. We ignore here the trivial cases $n=1,2$. Hadamard gave the bound $n^{n/2}$ for $g(n)$. This bound can be attained only if $n$ is a multiple of $4$. A matrix that attains it is called an {\it Hadamard matrix}, and it is an outstanding conjecture that one exists for any multiple of $4$. If $n$ is not a multiple of $4$, $g(n)$ is not in general known, but tighter bounds exist. For $n=2 \mod 4$, Ehlich and independently Wojtas proved that
 $$g(n)\leq (2n-2)(n-2)^{\frac{1}{2}n-1}\quad \mbox{(Ehlich-Wojtas' bound)}.$$

There is a companion theory for matrices with $0$'s on the main diagonal and $\pm 1$ elsewhere.
Let $f(n)$ denote the maximum determinant of all $n\times n$ matrices   with elements $0$ on the main diagonal and $\pm 1$ elsewhere.
It is well-known that $f(n)\leq (n-1)^{n/2}$. This can be attained only when $n$ is even, and a matrix which does so is called a {\it conference matrix}.
For $n= 2\mod 4$ there is a necessary condition: $n-1$ must be a sum of two squares. When $n$ is multiple of 4 that no condition is needed and $M$ is a  conference matrix if and only if $M+I$ is a  Hadamard matrix (it is called skew Hadadamard matrix). We shall here be concerned with the case $n= 2\mod 4$, $n\neq 2$ and this will be implicitly assumed in what follows.

Let $f_k(n)$ denote the maximum determinant of all $n\times n$ skew matrices with elements $0$ on the main diagonal and $\pm 1$ elsewhere. Let $M$ denote a matrix of this kind. Then,

\newpage

\begin{theorem}\label{theorembound}
We have
\begin{equation}\label{e-wbound01}
f_k(n)\leq (2n-3)(n-3)^{\frac{1}{2}n-1}
\end{equation}
 and equality  holds if and only if
there exist a skew matrix $M$ with
\begin{equation}\label{idenecsdc}
M M^T=M^TM=\left[\begin{array}{cc} L & 0 \\ 0 & L \end{array}\right],
\end{equation}
where $L=(n-3)I+2J$.
\end{theorem}

%Let $f(n)$ denote the maximum determinant of all $n\times n$  skew-symmetric matrices with elements $0$ on the main diagonal and $\pm 1$ elsewhere.

\begin{proposition}\label{nceb01}
 Equality in (\ref{e-wbound01}) can only hold if  $2n-3=x^2$ where $x$ is an integer.
\end{proposition}

\begin{lemma}\cite{Coh89}\label{detfdp}
Let $h\geq 2$, $A=[a_{i,j}]$ denote an $h$ by $h$ positive-definite symmetric matrix with diagonal elements $a_{i,i}=n-1$ and with $|a_{i,j}|\geq 2$. Then
\begin{enumerate}
\item $det(A) \leq (n+2h-3)(n-3)^{h-1}.$
\item equality holds if and only if $A=\Sigma L \Sigma$ for some suitable $\Sigma$, where $\Sigma$ is a square diagonal matrix with all its diagonal elements $\pm 1$.
    \end{enumerate}
\end{lemma}

\noindent{\em Proof of Theorem \ref{theorembound}} $\quad$ This proof can be considered a slight modification of the proof of \cite[Theorem 1]{Coh89}. It boils down to check some specific arguments. We have included it here in order to make the paper as much self-content as possible.

Dealing with the easier part first, it follows that if $MM^T=M^TM$ is of the form stated, then indeed holds, by Lemma \ref{detfdp}.

Now suppose that $M$ is any $n\times n$ skew matrices with elements $0$ on the main diagonal and $\pm 1$ elsewhere of determinant $f_k(n)$.
Let $A=MM^T$. Then $A$ is positive definite and symmetric, $a_{i,i}=n-1$ and for each $i\neq j$, $a_{i,j}=\sum m_{i,k}m_{j,k}= 0 \mod 2$. Thus every element of $A$ is even, and hence vanishes or is numerically at least 2. If no element of $A$ were to vanish, then by Lemma  \ref{detfdp} we should have
$$\begin{array}{rl}
\left(f_k(n)\right)^2= det(A) & \leq \,(3n-3)(n-3)^{n-1}\\
                    & =\, (3n^2-12n+9)(n-3)^{n-2}\\
                    & < \, (4n^2-12n+9)(n-3)^{n-2}\\
                    & = \, (2n-3)^2(n-3)^{n-2},
                    \end{array}$$
and the result would follow.

Otherwise, at least one element of $A$ vanishes and in particular is divisible by $4$. Now for any $i,j,t$ three distinct integers all in the range 1 to $n$ inclusive,

$$\begin{array}{rl}
a_{i,i}+a_{i,j}+a_{j,t}+a_{t,i} & =\,\sum (m_{i,k}^2+m_{i,k}m_{j,k}+m_{j,k}m_{t,k}+m_{t,k}m_{i,k})\\
                    & = \sum (m_{i,k}+m_{j,k})(m_{i,k}+m_{t,k})\\
                    & = 2+ m_{i,j}m_{i,t}-m_{j,t}m_{i,j}+m_{i,t}m_{j,t}\\
                    & \quad+\sum_{k\neq i,j,t} (m_{i,k}+m_{j,k})(m_{i,k}+m_{t,k})\\
                    & =
                    \, 3 \mod 4,        \end{array}$$
%$$\begin{array}{rl}
%a_{i,i}+a_{i,j}+a_{j,t}+a_{t,i} & =\,\sum (m_{i,k}^2+m_{i,k}m_{j,k}+m_{j,k}m_{t,k}+m_{t,k}m_{i,k})\\
 %                   & = \sum (m_{i,k}+m_{j,k})(m_{i,k}+m_{t,k})\\
 %                   & = 2+ m_{i,j}m_{j,t}-m_{j,t}m_{i,j}+m_{i,t}m_{j,t}\\
 %                   & \quad+\sum_{k\neq i,j,t} (m_{i,k}+m_{j,k})(m_{i,k}+m_{t,k})\\
 %                   & = \left\{\begin{array}{cl}
 %                   \, 3 \mod 4) & i,j,t \,\mbox{all differents}\\
 %                   \, 2\mod 4) & i,j,t \, \mbox{only one different}\\
 %                   \, 0 \mod 4) & i=j=t,
 %                   \end{array}\right.
 %                                       \end{array}$$
and so since $a_{i,i}=n-1= 1 \mod 4$ we find that
\begin{equation}\label{equatiomodular1}
a_{i,j}+a_{j,t}+a_{t,i}=
2 \mod 4.
\end{equation}
%\begin{equation}\label{equatiomodular1}
%a_{i,j}+a_{j,t}+a_{t,i}=\left\{\begin{array}{cl}
%2 \mod 4) &  i,j,t \,\mbox{all differents}\\
%1\mod 4) & i,j,t \, \mbox{only one different}\\
%\, 3 \mod 4) & i=j=t,
%\end{array}\right.
%\end{equation}
Let us  point out that to get the identity in the third line above we have used the property of being skew for $M$. As we have seen for at least one pair $(i,j)$ with $i\neq j$, $a_{i,j}=0\mod 4$. By interchanging rows of $M$ if necessary we can arrange that $a_{1,n}= 0\mod 4$. Suppose that with this done, there are precisely $h$ elements in the first row divisible by 4. Then $1\leq k\leq n-1$ since $a_{1,1}=n-1=1 \mod 4$ and $a_{1,n}= 0 \mod 4$. Then without affecting the first or last rows of $A$, but permuting other rows as necessary, we can arrange the $h$ elements in the first row of $A$ which are divisible by $4$ occupy the last $h$ positions, that is, that $a_{1,i}= 2\mod 4$ for $1\leq i\leq n-k$ and that $a_{1,i}=0\mod 4$ for $n-h+1\leq i\leq n$. Now partition $A$ in the form
$$
A=\left[\begin{array}{cc}
B(n-h,n-h) & C(n-h,h)\\
D(h,n-h) & E(h,h)
\end{array}\right].
$$
It then follows without difficulty from (\ref{equatiomodular1}) that every  off-diagonal element of $B$ and of $E$ is congruent to
2 modulo 4, whereas every element of $C$ and $D$ is divisible by 4. By Fischer's inequality we find that $det (A)\leq det(B)\, det(E)$ with equality if and only if every element of $B$ and of $C=B^T$ vanishes. Hence by Lemma \ref{detfdp},
$$
\begin{array}{rl}
\left(f_k(n)\right)^2 & =\, det (A)\\
                    & \leq \, (n+2(n-h)-3)(n-3)^{n-h-1}(n+2h-3)(n-3)^{h-1}\\
                    & = \,(n-3)^{n-2}\left((2n-3)^2-(n-2h)^2\right)\\
                    & \leq \,(n-3)^{n-2} (2n-2)^2,        \end{array}$$
with equality in the last line if and only if $k=\frac{1}{2}n$. Thus we have proved the inequality of the Theorem, and using Lemma \ref{detfdp}, can have equality only if
$$A=\left[
\begin{array}{cc}
\Sigma_1 L \Sigma_1 & 0 \\
0 & \Sigma_2L \Sigma_2
\end{array}
\right].$$
Now replacing $M$ by $\Sigma M\Sigma$ where
$$\Sigma=\left[
\begin{array}{cc}
\Sigma_1  & 0 \\
0 &  \Sigma_2
\end{array}
\right]$$
changes $M$ to a skew matrix in the form required. Obviously, $M^TM=MM^T$. This concludes the proof.

\vspace{2mm}

\begin{lemma}\label{lemmaconmuta}
Let $M$ be a skew matrix satisfying (\ref{idenecsdc}). Then
$$M X= XM$$
where $X=\left[\begin{array}{cc}J_{n/2} & 0\\
0 & J_{n/2}\end{array}\right]$.
\end{lemma}
\begin{proof}
Now let $$M=\left[\begin{array}{cc}
B & C \\
D & E
\end{array}
\right].$$
Since $M M^T=M^T M$, it follows that $M$ and $MM^T$ commute and so $L$ commutes with each of $B$, $C$, $D$ and $E$. But
$L=(n-3)I+2J$ and so $J$ too commutes with each of $B$, $C$, $D$ and $E$. Therefore, we conclude that
$XM=MX$ as it was desired.
\end{proof}

%\noindent{\it Proof Remark 1.} Assuming the same hypothesis and notation that the lemma above, it holds:
% \begin{enumerate}
% \item $B$ and $E$ are skew matrices.
% \item $C=-D^T$.
%  \end{enumerate}
%  Let us observe that
%  $$[JB]_{ij}=\sum b_{k,j};\quad [BJ]_{ij}=\sum b_{i,k}$$
%  and so the sum of all the elements in any row of $B$ is independent of the row, and similarly for the columns. Since $J$ commutes with $B$, $[JB]_{ij}= [BJ]_{ij}$. Let this sum be  $b$, say and define $c$, $d$ and $e$ similarly. But $B$ and $E$ are skew thus $b=e=0$ and $C=-D^T$ then
%  $c=-d$. Without loos of generality we assume $c\geq 0$.

 % Now $L=BB^T+CC^T=DD^T+EE^T$, and $0=BD^T+CE^T$. Thus premultiplying by $J$ yields

\begin{definition}
 A $(-1,1)$-matrix $K$ of size $n$ is said to be of skew-symmetric (or skew) type if $K+K^T=2I$.
\end{definition}

\begin{theorem}\label{maxdeterdosmundos}
Let $K$ be a $(-1,1)$-matrix of skew type.
$$\mbox{det}\,K=(2n-2)(n-2)^{\frac{1}{2}n-1}\Leftrightarrow \mbox{det}\,(K-I)=(2n-3)(n-3)^{\frac{1}{2}n-1}.$$
\end{theorem}
\begin{proof}
 Let $A=KK^T$ and $A'=(K-I)(K-I)^T$.

Since $k_{i,j}=-k_{j,i}$ and $k_{i,i}=k_{j,j}=1$ then $a_{i,j}=\sum k_{i,l} k_{j,l}=\sum_{l\neq i,j} k_{i,l} k_{j,l}=a'_{i,j}$ and $a_{i,i}=1+a'_{i,i}$ so $$A=A'+I.$$

 On the other hand, it is well-known \cite[Theorem 1]{Coh89} that
 $$\mbox{det}\,K=(2n-2)(n-2)^{\frac{1}{2}n-1}\Leftrightarrow A=K K^T=P \left[\begin{array}{cc}
L & 0\\
0 & L
\end{array}\right] P +I$$ for some  $(-1,1,0)$-monomial matrix $P$.

Since $A-I=A'$, then
$$A'=P \left[\begin{array}{cc}
L & 0\\
0 & L
\end{array}\right] P$$
which is equivalent to
$$\mbox{det}\,(K-I)=(2n-3)(n-3)^{\frac{1}{2}n-1}.$$

\end{proof}

 \begin{definition}
 We say that $H$ is an $n$ by $n$ matrix of skew-EW type if  $H$ is a $(-1,1)$-matrix of skew type and $(H-I)(H-I)^T=\left[\begin{array}{cc}
L & 0\\
0 & L
\end{array}\right]$. (or equivalently, by Theorem  \ref{maxdeterdosmundos}, $H H^T= \left[\begin{array}{cc}
L & 0\\
0 & L
\end{array}\right] +I$).
\end{definition}

\noindent{\it Proof of Proposition \ref{nceb01}.}
 Let $M$ be an $n$ by $n$ skew matrix satisfying (\ref{idenecsdc}). Thus, by Theorem \ref{maxdeterdosmundos},
   $H=M+I$ is  matrix of skew-EW type. This means that $$HH^T=H^TH=\left[\begin{array}{cc}
L & 0\\
0 & L\end{array}\right]+I.$$
 Then, by \cite[Theorem 2]{Coh89},  $2n-2$ has a representation as the sum of two squares, $2n-2=x^2+y^2$ and
 $$H=\left[\begin{array}{cc}
 X & Y \\
 Z & W
 \end{array}\right]$$
in which every row sum and columns sum of each of $X$ and $W$ is $x$, each row sum and each column sum of $Y$ is $y$. Since $X$ is of skew type, we have that $x=1$.
Then $2n-3=y^2$.

\vspace{3mm}

Remark 1 claims that a matrix of skew-EW type cannot exist for orders 10, 18, 22, 30, 34, 38, 46, 50, 54, 58, 66, 70 ,74, 78, 82, 90.
The next lemma asserts  that  a necessary condition for the existence of $n\times n$ matrices of skew-EW type   is that $n-1$ must be a sum of two squares but no sufficient. No matrix of skew-EW type exists for $n=18$ although   $n-1=17=4^2+1^2$. Moreover, there is a  conference matrix.

\begin{lemma}
Let $n$ be  an integer and $n=2 \mod 4$.
If $2n-3=x^2$ for some integer $x$ then $n-1=y^2+z^2$ for some integers $y$ and $z$.

\end{lemma}
\begin{proof}
On one hand, we have
$$n=4k+2 \Rightarrow n-1=4k+1= \frac{8k+2}{2}.$$
On the other hand,
$2n-3=x^2$ substituting $n$, we have $8k+1=x^2$.
Hence
$$n-1=\frac{x^2+1}{2}.$$

Furthermore, since $8k+1$ is odd, $x$ is so. $$x=2t+1\Rightarrow x^2= 4t^2+4t+1=t^2+(t+1)^2.$$ Thus,
 $$ n-1=\frac{x^2+1}{2}=2t^2+2t+1=t^2+(t+1)^2.$$

\end{proof}

\subsection{Self-dual codes using $(-1,1)$-matrices of skew type}

We now make a sketch of a construction for  matrices of skew-EW type. Matrices of skew-EW type have been found using this technique.
This construction is made up of two steps. Firstly, it consists of looking for $(-1,1)$-matrices with determinant equal to  Ehlich-Wotjas' bound, in the case that $2n-3$ is a square. The search methods usually employed in the literature \cite{AAFG12,Orr05} for finding  these $(-1,1)$-matrices are based on
  decomposing  $\left[\begin{array}{cc} L' & 0 \\ 0 & L' \end{array}\right]$ where $L'=(n-2)I+2J$ as the product of a $(-1,1)$-matrix  and its transpose. To carry out this first step, we have used the cocyclic approach described in \cite[Algorithm 2 p. 869]{AAFG12}. Secondly, check if the matrices yielded by step 1 are equivalent to a skew matrix with $1$'s on the main diagonal. For this second step, we suggest the following procedure.

\vspace{2mm}

\noindent{ \bf Algorithm.}  Search for  matrices skew type equivalent to $M$.

\vspace{2mm}

\noindent{Input: a $(-1,1)$-matrix $M$ of order $n$.}
\newline
\noindent{Output: a   matrix $K$ of skew type equivalent to $M$, in the case that such matrix $K$   exists. }

\vspace{3mm}

\noindent{$\Omega\leftarrow \emptyset$}

%\newline

\noindent{${S}\leftarrow$  The symmetric group $S_n$  (that is, the group of all $n!$ permutations of $N=\{1,\ldots,n\}$).

%\newline

\noindent{while ${S}$ is not empty  $\,\{$ }
  \hspace*{0.5cm}\begin{enumerate}
  \item[] 1. Choose $\sigma$ in ${S}.$

  2. $\,{{S}}\leftarrow{{S}}\setminus\{\sigma\}$.

  3. if $m_{i,\sigma(i)}=1$, $k_{\sigma(i),j}=m_{i,j}$; otherwise $k_{\sigma(i),j}=-m_{i,j}$, where $\,1\leq i,j\leq n$. Let $K=[k_{\sigma(i),j}]$.

   4. Check whether $K$ is  matrix of skew type.
                      If not, go to 1; otherwise $\Omega\leftarrow K$.

   5. End while.

           \end{enumerate}
                   $\}$

\vspace{2mm}

\noindent{$\Omega$ }
%\end{algorithm}

\vspace{3mm}
\newpage

\noindent{\sc Verification:}

If $K$ and $M$ are equivalent then there exist $(-1,1,0)$-monomial matrices $P$ and $Q$ such that $PMQ^T=K$. If $K$ is of skew type then
$Q^T K Q$ is so. Let us observe that $ Q^T K Q= Q^TP M$ and $Q^TP$ is a $(-1,1,0)$-monomial matrix.

\vspace{3mm}

In the following, a technique for constructing self-dual codes using matrices of skew-EW type is described.

\begin{lemma}\label{lemmaskew}
Let $H$ be an $n\times n$ matrix of  skew-EW type  and  $X=\left[\begin{array}{cc}J_{n/2} & 0\\
0 & J_{n/2}\end{array}\right]$. Then the following identity holds
$$X(H-I)^T+(H-I)X=0.$$
\end{lemma}
\begin{proof}
It follows from $X(H-I)=(H-I)X$ (see Lemma \ref{lemmaconmuta}).
\end{proof}

\begin{proposition}\label{mgeneradora}
Let $H$ be a skew-EW matrix of order $n$ and suppose that there exist three elements $\alpha,\,\beta$ and $\gamma$ nonzero elements of $GF(p)$ such that
%$$ \left\{\begin{array}{l}
%a^2 + \frac{n}{2} b^2 + (n-1) c^2 =0 \,\mod p),\\[2mm]
%\frac{n}{2} b^2 + 2 c^2 =0 \,\mod p).
%\end{array}\right.$$
\begin{eqnarray*}\label{equationcodes}
\alpha^2 + \frac{n}{2} \beta^2 + (n-1) \gamma^2 & = &0 \,\mod p,\\
\frac{n}{2} \beta^2 + 2 \gamma^2 & = & 0 \,\mod p.
\end{eqnarray*}
Then the matrix
$$G=[\alpha I,\,\beta X+ \gamma (H-I)]$$
generates a self-dual code $C_\star(n)$ of length $2n$ and dimension $n$ over $GF(p)$.

\end{proposition}
\begin{proof}
The statement of this proposition follows by direct inspection from the fact that $GG^T=0\,\mod p$ and lemma \ref{lemmaskew}.
\end{proof}

An upper bound on the minimum weight of any of these codes is given in the following proposition.
\begin{proposition}\label{bounddistancest}
Let $H$ be a skew-EW matrix of order $n$. Let $C_\star(n)$ be a self-dual code over $GF(p)$ constructed from $H$ using Proposition \ref{mgeneradora}. Then
$$d(C_\star(n))\leq \left\{\begin{array}{ll}
3, & n=2\\
\frac{n}{2}+2, & n\geq 6.
\end{array}
\right.$$
\end{proposition}

\begin{proof}
 Taking  two rows of $H-I$. Denoted by $r_i$ and $r_j$. Both of them chose either from the 1st to $\frac{n}{2}$-th of the rows of $H-I$ or  from $\frac{n}{2}+1$-th to $n$-th. For being $H$ a matrix of skew-EW type then $r_i\cdot r_j=2$. Therefore $wt(r_i-r_j)=\frac{n}{2}$.
\end{proof}

Note that the bound for (symmetric) conference matrices is \cite{AG01}

$$d(C(n,n-1))\leq \left\{\begin{array}{ll}
3, & n=2\\
\frac{n}{2}+3, & n\geq 6.
\end{array}
\right.$$

We remember that the gaps in  Table 1 have motivated our investigation of self-dual codes constructed from  matrices of skew-EW type. These are given in the remainder of this section.

\begin{center}\begin{table}
\caption{Self-dual codes from $(-1,1)$-matrices of skew type of orders $n=6,14$ and $26$.}
$$\begin{array}{c|c|c|c|c} \hline  & N=12 &  N=28  & N=52 \\
\hline \begin{array}{c} p \\
\hline 3 \\ 5 \\ 7 \\ 23
\end{array}
 &
 \begin{array}{cccc}\quad \alpha & \quad \beta \quad & d & \quad d_b\quad \\
\hline   &  &  & \\
 &  &  & \\
2 &   2 & 5 & 6\\
 &  &  &
\end{array}
&
\begin{array}{cccc}\quad \alpha & \quad \beta \quad & d & \quad d_b\quad \\
\hline  1 & 1 & 9 & 9 \\
 2& 2 & 8 & \quad 10 \,- \,12\\
 &    &  & \\
9 & 7 & 9 &
\end{array}
&
\begin{array}{cccc}\quad \alpha & \quad \beta \quad & d & \quad d_b\quad \\
\hline  1 & 1 & 12 & \quad 15 \\
 &  &  & \\
 &    &  & \\
 &  &  &
\end{array}\\
 \hline \end{array}$$\end{table}\end{center}

%\vspace{2mm}

\noindent{ \hspace{-0.3 cm} A. $\,\,[12,6]$ Codes}

\vspace{1.2mm}

There is only one matrix of skew-EW type (up to equivalence) with $n=6$, and this has a form given by

$$\left[\begin{array}{cccccc}
1 & 1 & - & 1 & 1 & 1 \\
- & 1 & 1 & 1 & 1 & 1 \\
1 & - & 1 & 1 & 1 & 1 \\
- & - & - & 1 & - & 1 \\
- & - & - & 1 & 1 & - \\
- & - & - & - & 1 & 1
\end{array}\right].$$
%where $-$ denotes $-1$.
From the bound in Proposition \ref{bounddistancest}, $d(C_\star(6))\leq 5$. This bound is met by $C_\star(6)$ for $p=7$ if $\alpha=\beta=2$ and $\gamma=1$. So we got a self-dual code near to optimal ($d_b=6$), but a code with $d=6$ exists \cite{AG01}. We have used Magma, a computer algebra system \cite{BC02} for symbolic computation developed at the University of Sydney, to compute the distance of the self-dual codes. Actually, we computed the distance of the derived self-dual codes for each possible solution of the diophantine system of Proposition \ref{mgeneradora}. In this case, there are no new solutions. So, we cannot fill in any gaps in the rows of Table 1. %We have followed the same pattern in the next cases.

\vspace{2mm}

\newpage

\noindent{\hspace{-0.3 cm} B. $\,\,[28,14]$ Codes}

\vspace{1.2mm}

There is only one matrix of skew-EW type with $n=14$, and this has a form given by
$$\left[\begin{array}{cccccccccccccc}
1 & 1& -& -& 1& -& 1& 1& -& -& -& -& -& -\\
-& 1& 1& 1&
-& -& 1& -& -& 1& -& -& -& -\\
1& -& 1& 1& 1& -& -& -&
-& -& 1& -& -& -\\
1& -& -& 1& -& 1& 1& -& -& -& -& -& \
-& 1\\
-& 1& -& 1& 1& 1& -& -& -& -& -& 1& -& -\\
1& 1&
1& -& -& 1& -& -& 1& -& -& -& -& -\\
 -& -& 1& -& 1& 1&
1& -& -& -& -& -& 1& -\\
 -& 1& 1& 1& 1& 1& 1& 1& 1& -& 1&
-& -& 1\\
 1& 1& 1& 1& 1& -& 1& -& 1& -& -& 1& 1& 1\\
 1& -&
1& 1& 1& 1& 1& 1& 1& 1& -& 1& -& -\\
 1& 1& -& 1& 1& 1& 1& -& 1&
1& 1& -& 1& -\\
1& 1& 1& 1& -& 1& 1& 1& -& -& 1& 1& 1& -\\
1&
1& 1& 1& 1& 1& -& 1& -& 1& -& -& 1& 1\\
1& 1& 1& -& 1& 1& 1&
-& -& 1& 1& 1& -& 1
\end{array}\right].$$

The bound for $C_\star(14)$ is $d\leq 9$. This is met by self-dual codes $C_\star(14)$ with $p=3$ and  $23$. For $p=3$ , the code $C_\star(14)$ is optimal (its distance is the highest possible). For $p=5$ there are two inequivalent codes, both with distance $d=8$. The first with parameters $\alpha=\beta=2$ and $\gamma=1$ and the second with parameters $\alpha=\beta=1$ and $\gamma=2$.

A solution with $\gamma=1$ of the diophantine system for each $p\leq 23$ for which a self-dual code exists is given in Table 2. Let us point out that for $n=14$ all these solutions are new in the sense, that they have filled   gaps in the rows of Table 1.

\vspace{2mm}

\noindent{\hspace{-0.3 cm} C. $\,\,[52,26]$ Codes}

\vspace{1.2mm}

For what we know, there is only one known  matrix of skew-EW type with $n=26$, and it is posted in {\em http://personal.us.es/armario/articulos/c26.pdf}. The bound for $C_\star(26)$ is $d\leq 15$. For $p=3$,  $\alpha=\beta=\gamma=1$ gives $d(C_\star(26))=12$.

Length 40 was the limit of the search for minimum distance gave in \cite{AG01}.

\section{Conclusion}
In this paper we present a method for constructing self-dual codes. The conditions required, for the generated codes to be self-dual, are proved. Our technique presents a strong connection with the  methods given in \cite[Proposition 2.1 and 2.2]{AG01} for conference matrices. Moreover, for $N=28$ the results provided by our method and  Proposition \ref{sdcsw} complement   each other (see table 1 and 2).

For $n=86$ no conference matrix is  known. But $2*86-3=13^2$  and so a matrix of skew-EW type could exist. If a matrix of skew-EW type for $n=86$ were found then we would have self-dual $[172,86]$-codes for the first time using these techniques. This problem is left open and it is our next goal.

%\subsection{Acknowledgements}

 \subsection{Acknowledgements}
This work has been partially supported by the research projects FQM-016 and P07-FQM-02980 from JJAA and MTM2008-06578 from MIC\-INN (Spain) and FEDER (European Union).
 We thank  Kristeen Cheng for her reading of the manuscript.

%%%%%%%%%%%%%%%%%%%%%%%%%%%%%%%%%%%%%%%%%%%%%%%%%%%%%%%%%%%%%


\begin{thebibliography}{1}




\bibitem{AAFG12}

 \'Alvarez V.,  Armario J.A.,  Frau M.D.,  Gudiel F.:
\newblock The maximal determinant of cocyclic $(-1,1)$-matrices over $D_{2t}$.
\newblock{ Linear Algebra Appl.} {\bf 436}, 858--873 (2012).


\bibitem{AG01}

 Arasu K.T.,  Gulliver T.A.:
\newblock Self-dual codes over ${\bf F}_p$ and weighing matrices.
\newblock{ IEEE Trans. Inform. Theory} {\bf 47}, 2051--2055 (2001).

\bibitem{BC02}

 Bosma W.,  Cannon J.J.:
\newblock { Handbook of Magma Functions.}
\newblock{ 2.9 edition,} School of Mathematics and Statistics, University of Sydney (2002).



\bibitem{Coh89}

 Cohn J.H.E.:
\newblock On determinants with elements $\pm 1$, II.
\newblock{ Bull. London Math. Soc. {\bf 21}, 36--42 (1989).}

\bibitem{CD96}

 Colbourn C.J.,  Dinitz J.H.:
\newblock{The CRC Handbook of combinatorial Designs.}
\newblock Boca Raton, FL: CRC, (1996).


\bibitem{Gab2013}

 Gaborit P.:
\newblock {\em Table of Self-Dual codes.}\newline
\newblock http://www.unilim.fr/pages$\_$perso/philippe.gaborit/SD/SelfDualCodes.htm (2013). Accessed 26 September 2013.

\bibitem{GH99}

 Gulliver T.A., Harada M.:
\newblock New optimal self-dual codes over $GF(7)$.
\newblock{ Graphs and Combin.} {\bf 15}, 175--186 (1999).


\bibitem{GS79}

 Geramita A. V.,  Seberry J.:
\newblock{ Orthogonal Designs: Quadratic Forms and Hadamard Matrices.}
\newblock Marcel Dekker, New York, Basel, (1979).

\bibitem{How80}

 Howard E.:
\newblock{ Elementary Matrix Theory.}
\newblock Dover Publications, (1980).

\bibitem{KS09}

 Koukouvinos C.,  Simos D.E.:
\newblock Construction of new self-dual codes over $GF(5)$ using skew-Hadamard matrices.
\newblock{ Adv. Math. Commun.} {\bf 3}, 251--263 (2009).


\bibitem{Orr05}

 Orrick W. P.:
\newblock The maximal $\{-1,1\}$-determinant of order 15.
\newblock{ Metrica} {\bf 62}, 195--219 (2005).
\end{thebibliography}
\end{document}